%% file: main.tex
\documentclass[sigconf,10pt]{acmart}
\usepackage{tikz}
\usepackage{amsmath}
\usepackage{booktabs}
\usepackage[english]{babel}
\usepackage{blindtext}
\usepackage{pifont}
\usepackage{tcolorbox}
\usepackage{xcolor}
\usepackage{enumitem}
\usepackage{siunitx}
\usepackage{graphicx}
\usepackage{subcaption}
\usepackage{multirow}
\usepackage{multicol}
\usepackage{colortbl}
\usepackage{siunitx}
\usepackage{hyperref}
\usepackage{float}
\usepackage[ruled,vlined, linesnumbered]{algorithm2e}
\usepackage{listings}

\usepackage{algpseudocode}
\usepackage{makecell}
\usepackage{graphicx}
\usepackage{url}

\usepackage{amsthm}

\theoremstyle{definition}
\newtheorem{definition}{Definition}

\definecolor{LineNumberColor}{rgb}{0,0,1}

\newcommand{\systemname}{Tommy}

\newenvironment{parafont}{\fontfamily{ptm}\selectfont}{}
\newcommand{\Para}[1]{\vspace{2pt}\noindent\begin{parafont}\textbf{\textit{#1}}\end{parafont}}

\keywords{Fairness, Ordering, Sequencing, Clock Synchronization, Probabilistic Ordering}

\begin{document}

\title{Beyond Lamport, Towards Probabilistic Fair Ordering}

\author{
\rm{\Large Muhammad Haseeb$^{\text{\scriptsize[n]}}$ \enskip
    Jinkun Geng$^{\text{\scriptsize[n][s]}}$ \enskip
    Radhika Mittal$^{\text{\scriptsize[u]}}$ \enskip
    Aurojit Panda$^{\text{\scriptsize[n]}}$ \enskip  \\
    }
\rm{
    \Large 
 Srinivas Narayana$^{\text{\scriptsize[r]}}$ \enskip
 Anirudh Sivaraman$^{\text{\scriptsize[n]}}$ \enskip
    }
\\
  {\Large
  $^{\text{\scriptsize[n]}}$\textit{New York University}\enskip 
  $^{\text{\scriptsize[s]}}$\textit{Stony Brook University}\enskip 
  $^{\text{\scriptsize[r]}}$\textit{Rutgers University}\enskip 
  $^{\text{\scriptsize[u]}}$\textit{UIUC}\enskip}
}

\renewcommand{\shortauthors}{Muhammad Haseeb, Jinkun Geng, Radhika Mittal,\\ Aurojit Panda, Srinivas Narayana, Anirudh Sivaraman}

\begin{abstract}
    \input{abstract}
\end{abstract}

\ccsdesc[500]{Networks~Application layer protocols}
\ccsdesc[500]{Mathematics of computing~Probability and statistics}

\acmYear{2025}\copyrightyear{2025}
\setcopyright{acmlicensed}
\acmConference[HotNets '25]{The 24th ACM Workshop on Hot Topics in Networks}{November 17--18, 2025}{College Park, MD, USA}
\acmBooktitle{The 24th ACM Workshop on Hot Topics in Networks (HotNets '25), November 17--18, 2025, College Park, MD, USA}
\acmDOI{10.1145/3772356.3772401}
\acmISBN{979-8-4007-2280-6/25/11}

\maketitle

\input{introduction}
\input{relatedwork}
\input{new_apps}
\input{system_model}
\input{overlap_analysis}
\input{batching}
\input{streaming}
\input{evals}

\input{next}

\input{conclusion}
\input{acks}

\bibliography{main}

\input{appendix}
\end{document}

%% file: abstract.tex
A growing class of applications demands \emph{fair ordering} of events, which ensures that events generated earlier are processed before later events. However, achieving such sequencing is challenging due to the inherent errors in clock synchronization: two events at two clients generated close together may have timestamps that cannot be compared confidently. We advocate for an approach that embraces, rather than eliminates, clock synchronization errors. Instead of attempting to remove the error from a timestamp,  \systemname{}, our proposed system, leverages a statistical model to compare two noisy timestamps probabilistically by learning per-clock synchronization error distributions. 
Our preliminary statistical model computes the probability that one event precedes another by only relying on local clocks of clients. This serves as a foundation for a new relation: \emph{likely-happened-before} denoted by $\xrightarrow{p}$ where $p$ represents the probability that an event happened before another. The $\xrightarrow{p}$ relation provides a basis for ordering multiple events which are otherwise considered \emph{concurrent} by Lamport's \emph{happened-before} ($\rightarrow$) relation. We highlight various related challenges including the intransitivity of the $\xrightarrow{p}$ relation as opposed to the transitive $\rightarrow$ relation. We outline several research directions: online fair sequencing, stochastically fair total ordering, and handling byzantine clients. 

%% file: introduction.tex
\section{Introduction}

Sequencers play a pivotal role in distributed systems, providing a mechanism to impose a total order on events occuring potentially at different locations. They are essential components in several fundamental protocols, such as consensus and concurrency control. In consensus protocols (e.g., Paxos~\cite{paxos} and Raft~\cite{raft}), the leader node serves as the sequencer for deciding a total order, as well as an orchestrator for achieving agreement on the total order. More recently, network-based sequencers have been introduced to offload some of the complexity from these protocols. Systems such as NOPaxos~\cite{nopaxos}, Hydra~\cite{hydra}, and Eris~\cite{eris_dan} decouple sequencing from the rest of the functionality, proposing dedicated sequencers to improve overall system efficiency.

At its core, the function of a sequencer is simple: assign ranks to incoming messages, thereby establishing a total order for processing the messages. This ranking is typically independent of when a message was originally generated. Instead, it is assigned based on the order in which it is \emph{observed} by a server/sequencer (i.e., FIFO sequencer). In most traditional applications, this FIFO approach suffices, as the system only requires \textit{some} ordering, even if arbitrary. 

\begin{figure}[!t]
    \centering
    \includegraphics[width=0.5\textwidth]{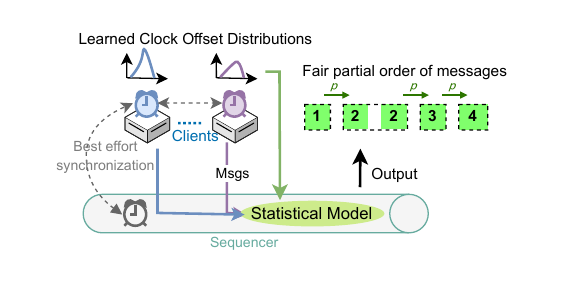}
    \vspace{-1cm}
    \caption{The sequencer, \systemname{}, uses clock offset distributions and noisy timestamps of messages to achieve a fair ordering of messages via a statistical model. Messages whose order cannot be confidently determined become part of the same output batch. }
    \label{fig:sys_model}
    \vspace{-0.5cm}
\end{figure}

We make a case for fair ordering which, unlike FIFO ordering, requires that \emph{an earlier generated event is ordered before a later generated event}. The FIFO order could be naturally closer to the fair order if the time between generation of every two events is large enough that arbitrary network delays do not obscure the order of events. However, there is a rise in applications in which a large volume of events is generated close together. These applications demand a sequencing mechanism that explicitly aligns the ordering of events with the timestamps at which the events are generated. It is particularly prominent in financial exchanges, ad exchanges, and other competitive systems~\cite{adex, cloudex, dbo, shoebot1, shoebot2, bot3, aerial_marketplace}, where fairness is paramount, we call such applications \emph{auction-apps}. In such applications, millions of events by hundreds of clients are generated within a very small window of time upon some sensitive event, for example, in financial exchanges some event leading to market volatility may be broadcasted to all the clients simultaneously~\cite{cloudex, dbo, jasper}, eliciting a large volume of responses by the clients. In these settings, ensuring that an earlier-generated message is ranked lower (processed sooner) than a later-generated one is crucial for maintaining fairness among participants. It is because of such fairness requirements and lack of fair ordering primitives, that exchanges today are built in private data-centers and not on a general purpose networking fabric e.g., that of the public cloud as seen by its tenants. 

\Para{Recent Efforts:} Recently the community has alluded to such ordering in the context of auction-apps, but either the solutions are (i) impractical~\cite{cloudex, jasper} due to strong assumptions (e.g., near-perfect clock synchronization), or (ii) not generally reusable because of being coupled with the intricacies of a particular application~\cite{dbo}. We define a general mechanism for achieving fair ordering as a \emph{fair sequencer}: a sequencer that guarantees that an earlier event is ranked lower (i.e., processed sooner) than a later event with a high probability. 

\Para{Classical Context:} Lamport's seminal work on ordering of events~\cite{leslie_ordering} introduces the \emph{happened-before} ($\rightarrow$) relationship. If two events $a$ and $b$ are causally related i.e., $a$ causes $b$, then they can be ordered i.e., $a \rightarrow b$. The relation $\rightarrow$ is a transitive relation so a set of related events can be partially ordered. Two concurrent events, i.e., for whom a causal relationship cannot be determined are left unordered i.e., $a \nrightarrow b$ and $b \nrightarrow a$. We are precisely interested in ordering such concurrent events; a hard feat in its general essence as we establish in this paper, but very much needed for fair ordering. 

\Para{Fundamental Challenge:} Ideal fair ordering requires perfect clock synchronization so that two timestamped-events (from two different clients) can be ordered correctly even if the network reorders them. Perfect clock-synchronization is impossible to achieve in asynchronous or bounded-synchronous networks~\cite{limits_on_clock_sync, lundelius_clock} due to fundamental uncertainty around link delays. It is impossible to synchronize clocks of $n$ processes any more closely than $u(1 - 1/n)$ where $u$ represents the uncertainty in the link delays~\cite{lundelius_clock}. This impossibility of clock synchronization makes it challenging to achieve fair ordering even if all parties are trusted~\cite{dbo}.

\Para{An Approximate Solution and When It Fails:} In a constrained setting where the time resolution of interest is significantly coarser than the clock synchronization errors, the fair sequencer can be implemented by a straightforward algorithm as clock errors can be effectively ignored: by waiting for at least one message from every client and then releasing the message with the smallest timestamp, iteratively. This algorithm achieves a fair total order, provided in-order delivery of messages per client. This approach is practical in environments where all client VMs and the sequencer reside within a single data center, as clock synchronization errors can be reduced to nanoseconds~\cite{huygens}, making it practical for systems operating at microsecond or higher time resolutions. However, when the required resolution is finer or clock synchronization errors become pronounced, such as in multi-data center deployments where the errors easily reach tens of microseconds, this approach is insufficient. To address these broader challenges, we call for a generally fair sequencer.

\Para{A Research Vision and Associated Challenges:} 
We advocate leveraging the insight that two \emph{local} timestamps from two clients can be compared if the clock offsets distributions of the clients are known. A client can learn its distribution of clock \emph{offsets} (w.r.t. the sequencer's clock), for example, by accumulating synchronization probes\footnote{A synchronization probe is a packet sent by a clock synchronization protocol from one client to the other to find and correct any clock offset.} from any clock synchronization protocol. The learned offsets' distributions are shared with the sequencer, enabling a comparison of two local timestamps. Figure~\ref{fig:sys_model} sketches a possible system architecture. Based on this ability, we introduce a new relation: \emph{likely-happened-before}, $\xrightarrow{p}$ where $p$ denotes the probability i.e, in $x \xrightarrow{p} y$, $x$ happened before $y$ with probability $p$. Similar to how $\rightarrow$ relation is used for defining a partial order on events, the $\xrightarrow{p}$ relation can be used to provide a \emph{fair} partial order. However, as the $\xrightarrow{p}$ relation is probabilistic, ordering \emph{all} concurrent events with high confidence may not always be possible. Hence, only a partial order is expected. It is important to minimize such instances of \textit{non-ordering} as otherwise a trivial solution is to leave all events as unordered. 
This ordering based on $\xrightarrow{p}$ constitutes fair ordering. 

There are two main challenges in using the $\xrightarrow{p}$ relation to achieve fair ordering: (i) unlike the $\rightarrow$ relation, the $\xrightarrow{p}$ relation is not necessarily transitive, so using it to order more than two events is non-trivial and, (ii) finding the probability \emph{p} for constructing $\xrightarrow{p}$ relations. We later present a preliminary statistical model to calculate \emph{p}. Once \emph{p} is known, it can be used to obtain an ordering which has high confidence~(\S\ref{batching}). 

\Para{Intransitivity and Ordering of Multiple Events: } It is possible for the probability of event A preceding event B to be high, the probability of B preceding C to be high, and yet the probability of C preceding A to also be high. In a similar vein, an ordinary cat may prefer fish to meat, meat to milk and milk to fish, in exhibiting cyclic ordering. This renders $\xrightarrow{p}$ not necessarily a transitive relation, hindering us from defining an order on the events from pairwise relations. We later discuss a solution for handling such intransitivity, while also presenting a sequencer for the case where probabilities are transitive. Transitivity exists for some \emph{nicely shaped} distributions like Gaussian distributions (proof in Appendix~\ref{app:proof_guass}) but may not hold for arbitrary distributions (e.g.,~\cite{efron_dice}).

Furthermore, online sequencing is an equally challenging problem as sequencing a given set of events primarily because of (i) network asynchrony and, (ii) figuring out whether some future events may need the same or lower rank than some given events. We later discuss a direction for achieving online sequencing. We prototype our statistical approach, \systemname{}, and present simulation results demonstrating its effectiveness compared to a naive TrueTime (Spanner) based baseline~\cite{spanner}. We highlight a range of research directions enabled by our approach --potentially culminating in a novel sequencing primitive that supports a broad class of emerging applications atop general-purpose networking infrastructure.

%% file: relatedwork.tex
\section{Related Work and Motivation}

\Para{Cloud Exchanges:} Recent proposals for cloud-hosted financial exchanges~\cite{cloudex, jasper, dbo} deal with the same sequencing problem. However these systems either simplify the problem by making strong assumptions like negligible clock synchronization errors or they reduce the scope of the problem by limiting what kinds of events are possible. Figure~\ref{fig:sequencer_b} shows a \textbf{W}aits \textbf{F}or \textbf{O}ne (WFO) sequencer which waits for one message from all clients and releases the one with the smallest timestamp, iteratively. This sequencer is employed by~\cite{onyx} and works as long as the clock synchronization errors are small enough to be ignored so that the timestamps on the messages can be considered representing a global-clock time. 

\Para{On-Prem Exchanges:} On-premise exchanges engineer their infrastructure for fair ordering: connecting all clients to the server using equal length wires and employing low jitter switches (e.g., L1 switches~\cite{l1_switch}). In such a setting, the server can process messages in the order of their arrival which would be equivalent to ordering them on their generation timestamps (Figure~\ref{fig:sequencer_a}). However, such a sequencer can only be deployed by modifying the underlying infrastructure. \systemname{}, our proposal, is a solution that does not make such assumptions or require special infrastructure (Figure~\ref{fig:sequencer_c}). 

\begin{figure}[!t]
    \centering
    \begin{minipage}[t]{0.15\textwidth}
        \centering
        \includegraphics[width=1.3\textwidth]{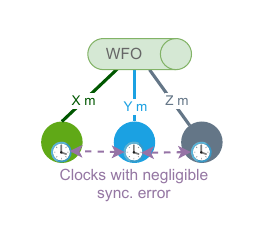}
        \vspace{-1.3cm}
        \captionof{figure}{\textmd{Fair if clocks are perfectly synchronized.}}
        \label{fig:sequencer_b}
    \end{minipage}
    \hfill
    \begin{minipage}[t]{0.15\textwidth}
        \centering
        \includegraphics[width=1.32\textwidth]{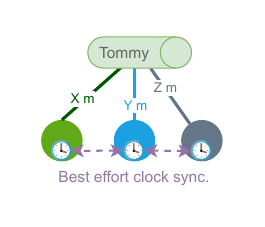}
        \vspace{-1.3cm}
        \captionof{figure}{\textmd{Fair w/o constraints but probabilistically.}}
        \label{fig:sequencer_c}
    \end{minipage}
    \hfill
    \begin{minipage}[t]{0.15\textwidth}
        \centering
        \includegraphics[width=1.1\textwidth]{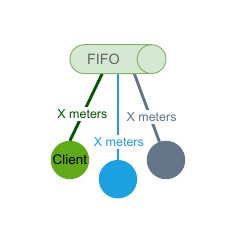}
        \vspace{-1.3cm}
        \captionof{figure}{\textmd{Fair if all wires are of equal length.}}
        \label{fig:sequencer_a}
    \end{minipage}

\end{figure}

\Para{Departing from Arbitrary Ordering}: Pompe~\cite{byantine_ordered_consensus} proposes departing from an arbitrary total order and instead allowing the nodes of a Replicated State Machine to present hints about their desired ordering of events. However, Pompe is fundamentally different from \systemname{} in its goals. Pompe focuses on the question of how to keep a subset of replicas from influencing the ordering of events unilaterally. It has applications in settings e.g., blockchains, where Byzantine failures are possible. \systemname{}, on the other hand, focuses on whether and how an ordering of events can be achieved which reflects the true order of event occurrences. 


%% file: new_apps.tex
\Para{Our Motivation: } Our motivation stems from the efforts around migrating financial exchanges to the public cloud. Financial exchanges have traditionally been built in private data centers or colocation facilities, where the physical network is engineered to provide fairness guarantees. This eliminates the need for a fair sequencer in such environments. However, a recent wave of research~\cite{cloudex, dbo, jasper, onyx} exploring the migration of financial exchanges to the public cloud has created a demand for new networking primitives. One such primitive, briefly mentioned in Onyx~\cite{onyx}, is a sequencer for fair total ordering. The design of Onyx assumes that clock synchronization errors are significantly smaller than the time resolution of interest, allowing it to disregard clock variability. However, we observe that this assumption does not hold if the system is deployed across multiple cloud regions where the clock synchronization errors can be significantly higher (e.g., in the order of milliseconds~\cite{tiga}), necessitating a more generalized fair sequencer.

\Para{Auction-apps:} Beyond financial exchanges, many applications can benefit from such a sequencer, including ad exchanges and competitive marketplaces. \emph{An application involving a shared state among multiple clients, where any particular order of writes can be advantageous/disadvantageous for some clients i.e., clients may compete to write earlier than others, is a candidate for fair sequencing.} We call such applications \emph{auction-apps}. The rise of competitive marketplaces~\cite{adex, cloudex, dbo, shoebot1, shoebot2, bot3, aerial_marketplace} and our discussion with experts demonstrate that such applications are becoming ubiquitous. 

\Para{Fairness: } We use the term fairness differently from the typical networking/scheduling notions of fairness i.e., Jain's index~\cite{jainsfairness}, CFS~\cite{cfs} or throughput-centric fairness. We define fairness in sequencing as follows:

\begin{definition}[Fair Sequencing]
Messages generated by the clients should be seen by a server in the same order as their generation is observed by an omniscient observer.\footnote{An omniscient observer has access to a global clock with infinite resolution and has instantaneous knowledge of all events. It serves as an idealized scheme to compare against.}
\label{def:fairsequencing}
\end{definition}

Other notions of fairness in sequencing may also be possible, but we only focus on Definition~\ref{def:fairsequencing}. Furthermore, in practice, the timestamp (w.r.t the local clock) at which an event actually occurs and the timestamp (w.r.t the local clock) that is associated with the event generation by a client can have non-zero difference because of some \emph{latency} between the application and the clock. For the sake of this position paper, we assume the difference is negligible.

%% file: system_model.tex
\section{Preliminary Design for \systemname{}}

Each client's clock may have some error w.r.t. the sequencer's clock due to imperfect clock synchronization.\footnote{Synchronizing clients' clocks with the sequencer's clock is sufficient as opposed to synchronizing clients' and sequencer's clock with a global clock.} The sequencer, \systemname{}, receives messages from clients with timestamps attached, attempts to order them and form batches ($B_i, B_j,..$). All messages within a batch $B_i$ will have a rank $i$ where successive batches have higher ranks. Ideally, if message $a$ is created before message $b$ according to the global-clock, then the rank of the batch containing $a$ should be smaller than the rank of the batch containing $b$. If two timestamps cannot be ordered confidently, then the corresponding messages should be part of the same batch. The challenge is to come up with the batches that maximizes fairness: the more batches we make, the better fairness we achieve.\footnote{Assuming no two events occur at the same instant.}

We decompose the above problem into two steps: (i) finding the probability of one message preceding another message (\S\ref{ordering_prob}, \S\ref{arb_dists}) to construct the $\xrightarrow{p}$ relation and, (ii) using the pairwise relationships to get ordered batches (\S\ref{batching}) that provides a fair partial order on all messages. We assume all messages are present at the sequencer before it starts sequencing. Later in~\S\ref{sec:streaming}, we lift this assumption. The preliminary system does not handle the case where clock offset distributions lead to intransitive probabilities on the order of events, but we discuss a direction for a possible solution. 

\subsection{System Model}

Each client submits a message to the sequencer and attaches the current timestamp from its local clock. A message $i$ has timestamp $T_i$. 
However, due to clock synchronization errors, the true timestamp of the message (from the sequencer's perspective) is:
\(T_i^* = T_i + \theta_i\)
where $\theta_i$ represents the clock offset of a client (w.r.t the sequencer's clock) at the exact moment when the message $i$ is generated. The offset $\theta_i$ is unknown but follows probability distribution $f_{\theta_i}$.
The sequencer can observe $T_i$, not $T_i^*$. 

Different clients may have different distributions due to heterogeneous synchronization conditions (e.g., different temperature in different parts of a data center, asymmetric latency between clients). Each client learns their own distribution (by accumulating clock synchronization probes) and provides information about their distribution to the sequencer (\S\ref{sec:future}).

%% file: overlap_analysis.tex
\subsection{Ordering Probability}
\label{ordering_prob}

It is impossible to compute $T_i^*$ exactly but we can compare two timestamps $T_i^*$ and $T_j^*$ by only observing $T_i$ and $T_j$ using a probabilistic analysis that assumes the knowledge of clock offset distributions $f_{\theta_i}$ and $f_{\theta_j}$.

We analyze the probability that one event/message precedes another. This probability is called the \emph{preceding-probability}:
\[
    \mathbb{P}(T_i^* < T_j^* \mid T_i, T_j) = \mathbb{P}(T_i + \theta_i < T_j + \theta_j).
\]
Rearranging,
\[
    \mathbb{P}(T_i^* < T_j^* \mid T_i, T_j) = \mathbb{P}(\theta_j - \theta_i > T_i - T_j).
\]
Since $\theta_i$ and $\theta_j$ are random variables, their difference follows a new distribution:
\[
    \Delta \theta = \theta_j - \theta_i \sim f_{\Delta \theta}.
\]
Then the preceding-probability is given by:
\[
    \mathbb{P}(T_i^* < T_j^* \mid T_i, T_j) = \int_{T_i - T_j}^{\infty} f_{\Delta \theta} d\Delta.
\]

If two independent random variables follow Gaussian distributions, then their difference also follows a Gaussian distribution~\cite{Proakis1985ProbabilityRV}. Therefore, for independent Gaussian-distributed clock synchronization errors, $\Delta \theta$ would be Gaussian-distributed. In this case, the preceding-probability is simply \(\Phi \left( \frac{T_j - T_i + (\mu_i - \mu_j)}{\sqrt{\sigma_i^2 + \sigma_j^2}} \right),\) where $\Phi(x)$ is the standard normal CDF, and $\mu_i$ and $\sigma_i^2$ respectively represent the mean and variance of $f_{\theta_i}$. 

\subsection{Handling Arbitrary Distributions}
\label{arb_dists}

When the clock offsets \( \theta_i \) and \( \theta_j \) follow arbitrary distributions rather than Gaussian or when we are uncertain about the distribution of $\Delta \theta$, we may not have a well-known solution form. Such cases have been reported where although the clock-offsets data appear Gaussian-like, it shows a long tail and skewed behavior~\cite{kim2015modeling}. We must estimate the PDF $f_{\Delta \theta}$ for each pair of clients to compute the preceding probabilities to account for non-Gaussian behavior. 

\Para{Computing all $\Delta \theta$s to get $f_{\Delta \theta}$:} 
For each round of clock synchronization probes to the clients, a sequencer could gather \emph{all} the probes and calculate pairwise probe differences ($\Delta \theta$s) and learn their distribution ($f_{\Delta \theta}$) across several rounds. This is communication and computation intensive.
If a clock sync. protocol has a high probe frequency, it would increase the communication to sequencer as well. However, a simpler and efficient method exists, explained below. 

\Para{Clients learn their own $f_{\theta_i}$:} If clients learn their own offset (w.r.t. the sequencer's clock) distributions over several rounds of clock synchronization, they can share their respective distributions with the sequencer which could perform (pairwise) convolutions to estimate $f_{\Delta \theta}$ for each pair of clients. 

\Para{Convolution for finding the Probability Density Function (PDF):} The PDF of $\Delta \theta$$=\theta_j - \theta_i$ is given by the \emph{convolution} of the individual PDFs \( \theta_i \) and \( \theta_j \) i.e., $\displaystyle
f_{\Delta \theta}(\Delta) = \int_{-\infty}^{\infty} f_{\theta_j}(\xi) f_{\theta_i}(\xi - \Delta) d\xi.$ This approach requires less communication from the clients to the sequencer as clients merely send their respective learned distributions to the sequencer as opposed to sending all the clock synchronization probes. 

\Para{Optimizing the convolutions calculations:} The calculations of all pairwise convolutions at the sequencer can further be optimized by leveraging Fast Fourier Transform: convolution in the time domain is multiplication in the frequency domain. Instead of computing a convolution, we can (i) compute Fourier transforms of $f_{\theta_j}$ and $f_{-\theta_i}$, (ii) multiply them point-wise and, (iii) compute the inverse Fourier transform to get $f_{\Delta \theta}$. This process has log-linear time complexity if using FFT, as opposed to the quadratic complexity of convolution.

Once \( f_{\Delta \theta} \) is obtained, the preceding-probability is simply:
\(\displaystyle
\mathbb{P}(T_i^* < T_j^* \mid T_i, T_j) = \int_{T_i - T_j}^{\infty} f_{\Delta \theta}(\Delta) d\Delta.
\)
This framework supports arbitrary clock error models, making it robust for real-world environments. 

%% file: batching.tex
\subsection{Fair Ordering}
\label{batching}

Once we define the $\xrightarrow{p}$ relation, we can work towards ordering multiple events. We model each message as a node in a graph, where $\xrightarrow{p}$ denotes a directed edge with weight $p$. In our construction, there will be two edges between each pair of nodes; for every such pair, we discard the edge with the lower weight (assuming no ties). From the resultant graph, we can extract a linear ordering of events by finding a topological ordering. Questions remain whether a topological ordering exists or which topological ordering to select if multiple orderings are possible. 

Assuming clock offsets distributions that lead to transitivity for $\xrightarrow{p}$, the graph forms a \emph{transitive tournament}~\cite{Gass01061998}.\footnote{A directed graph with exactly one edge between every pair of nodes is called a tournament. A tournament in which the edge relation is transitive is called a transitive tournament.} Transitive tournaments have a unique Hamiltonian path, hence a unique topological ordering. So the problem simplifies in the case of transitivity. In Appendix~\ref{app:proof_guass}, we prove how Gaussian distributions always lead to the required transitivity. Appendix~\ref{app:fair_ordering_example} illustrates an example with several events and respective transitive preceding probabilities and how fair ordering is achieved. 

In the case of intransitivity of $\xrightarrow{p}$, the resulting graph could be cyclic so no topological ordering may be possible. We may need some transformation of the graph to enable extracting a (most probable) linear ordering. One option is to remove some edges that renders the graph acyclic. However, it would lead to unfairness towards some messages/clients. For example, for three events whose preceding probabilities form a cycle, we could remove one edge to get a linear ordering but the removal of the edge would lead to ignoring one preceding-probability in the final linear ordering. A notion of stochastic fairness could be introduced and every time a set of messages is processed, we remove some edges from the graph in a fashion that leads to fairness over the long run. However, finding the smallest set of edges whose removal would make a graph acyclic is an NP-hard problem. These aspects of fair ordering make the problem non-trivial under intransitivity, warranting further research. 

The extracted linear ordering from the graph, even under transitive probabilities, cannot be construed as a final ordering. $\xrightarrow{p}$ relations of some adjacent messages in the linear ordering have a $p$ just slightly above 0.5 while other may have a $p$ close to 1; so it cannot be considered fair with a reasonable confidence. We batch adjacent messages such that if $i \xrightarrow{p} j$ has $p > {threshold}$ then a batch boundary is created between $i$ and $j$, making $i$ and $j$ belong to two different batches. Finally, the first such batch is assigned a rank of 0 while successive batches get incremental ranks, yielding a fair ordering of messages. The messages which we cannot order confidently become part of the same batch; thus our ordering is partial and not total. The hyper-parameter ${Threshold}$ dictates the confidence of our ordering and needs to be selected carefully. 

A ${Threshold}$ closer to 1 creates fewer and bigger batches, while a ${Threshold}$ closer to 0.5 creates smaller and more batches. A higher value of ${Threshold}$ provides more confidence in the output ordering but may lead to more number of messages left as unordered i.e., part of the same batch. Ideally, each batch should be of size 1. Hence, maximizing fairness amounts to creating smaller batches. While maximizing correctness may require staying indifferent about the (concurrent) messages, i.e., making them part of the same batch as we can never be 100\% confident about ranking of batches. We leave the optimization of ${Threshold}$ as future work and currently use a value of 0.75 in the evaluation.  

Although we achieve partial ordering on the messages, it is a total ordering on the batches. The sequencer emits one batch at a time to an upstream application for further processing of the corresponding messages. 

%% file: streaming.tex
\subsection{Online Sequencing}
\label{sec:streaming}

The above discussion on ordering assumes that the sequencer has received all the messages that need to be sequenced. However, in practice, messages arrive as a stream, and the sequencer must operate in an online fashion. Crucially, the sequencer must ensure that once a batch of messages is \emph{emitted}, i.e., released after sequencing, no new message should arrive that either belongs in the same batch or demands a lower rank. 

\Para{Two main questions: } The challenge of online sequencing boils down to answering two key questions.
\textbf{Q1:} Given a batch of timestamps (of messages), what future timestamps might still need to be included in the current batch? 
\textbf{Q2:} How can we ensure that all messages with timestamp $t$ (or $\leq t$) have already arrived at the sequencer? Q1 arises due to clock synchronization errors --specifically, a client $c$ may have enough uncertainty in their local timestamps that messages from another client, with later timestamps, must be grouped with $c$’s messages. In such scenarios, although two messages $i, j$ from a client can be ordered w.r.t each other, they must belong to the same batch as a third \emph{high-uncertainty message $k$} from another client. This is required because $\mathbb{P}(T_i^* < T_k^*)$ as well as $\mathbb{P}(T_j^* < T_k^*)$ can both be very small. The second question reflects the challenges introduced by network asynchrony. Appendix~\ref{app:online_seq_example} illustrates online sequencing with an example.

\Para{Safe batch emission:} We hint at how the answer to Q1 can be extracted which is equivalent to calculating \emph{waiting-period} to safely emit a batch. The sequencer can safely emit a batch if no new message that needs a lower or equal rank arrives during this waiting period, otherwise a new waiting period is calculated accounting for the newly received messages. This could in theory lead to blocking the sequencer from emitting any messages if the arrival pattern of messages and the clock offsets distributions are set adversely. We have not tackled this yet.

A safe way to emit a batch is to calculate a future time $T^F_i$ for each message $i$ in the batch such that 
\[\mathbb{P}(T^*_i < T^F_i) > p_{\mathrm{safe}}\]
where $p_{\text{safe}}$ can be set to a high value to ensure enough confidence (e.g., 0.999). $T^F_i$ that respects the above constraint can be trivially and efficiently computed by a binary search on the future timestamps. 

The safe emission time for the entire batch becomes:
\[T_b = \max_{k} \big( T^F_k \big) \quad \forall k \in \text{batch}\]
The sequencer after finalizing a batch, will only emit it (i) once its clock reaches $T_b$ timestamp and, (ii) it has not received any further messages that should be part of the batch or deserve a lower rank. If new messages arrive before $T_b$ which violate (ii), then $T_b$ is extended accounting for the new messages. 
The parameter $p_{\text{safe}}$ presents a trade-off between latency of emitting a batch and certainty of fairness. 

\Para{Dealing with network asynchrony:} There are several directions for dealing with network asynchrony (for Q2). Assuming bounded asynchrony and waiting for sufficiently long enough is a common practice while~\cite{cloudex} studies the impact of waiting period on fairness. Another direction, applicable to \emph{auction-apps} is to assume the knowledge of a fixed number of clients. This simple knowledge is powerful in answering Q2. To ensure all messages generated before some timestamp $t$ have arrived, the sequencer simply waits for messages or heartbeats with timestamp greater than $t$ from \emph{all clients}. This works as long as the communication between each client and the sequencer happens through an ordered delivery channel (e.g., TCP connection). However, this design may cost liveness i.e., a failed client may halt the sequencer from emitting any messages.

%% file: evals.tex
\begin{figure}[!t]
    \centering
    \includegraphics[width=0.4\textwidth]{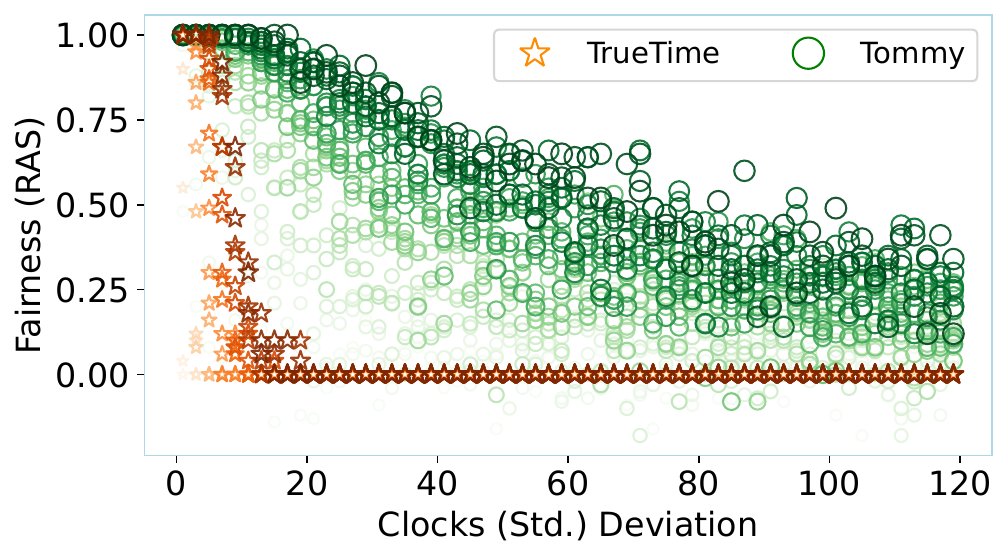}
    \caption{Tommy achieves fairer sequencing than TrueTime. Size of the marker (and color intensity) is proportional to the inter-messages gap across clients. }
    \label{fig:truthfulness2}
\end{figure}

\section{Evaluation}

We evaluate our statistical model using a simulator with 500 clients, each assigned a Gaussian clock offsets distribution, \(N(\mu,\sigma^{2})\). At message generation, a client reads the wall-clock time $t$, samples noise $\epsilon$ from the distribution, and tags the message with $T = t + \epsilon$. The sequencer receives all messages before ordering, i.e., we do not evaluate online sequencing. Ground-truth timestamps ($t$) are also collected for evaluation. Clients send their distributions and timestamped messages to the sequencer as shown in Figure~\ref{fig:sys_model}. We seed the clients with clock offsets distributions, instead of clients learning such distributions, so the following results are an upper-bound on the performance as the errors in estimating such distributions are not captured. 

For baseline, we emulate Spanner TrueTime~\cite{spanner}, where each message is assigned an uncertainty interval $[T - 3\sigma,\, T + 3\sigma]$, and overlapping intervals are assigned the same rank. 

We propose a metric, Rank Agreement Score (RAS): $+1$ for each correct ordered pair, $-1$ for incorrect, and $0$ for indifference i.e., for assigning same batch to a pair of messages. 

Figure~\ref{fig:truthfulness2} shows RAS (each point is the sum of RAS of all pairs of messages) for both approaches, with marker size (and color intensity) showing inter-messages gaps across clients. With low clock errors (lower x-axis), both systems perform comparably. \systemname{} outperforms (higher y-axis) TrueTime, when inter-messages gap decreases (marker size/color intensity decreases) and/or clock errors increase (higher x-axis). 
However, \systemname{}'s probabilistic nature can lead to negative RAS under high uncertainty/high clock errors, whereas TrueTime’s RAS remains 0 due to its conservative nature.

%% file: next.tex
\section{Future Research}
\label{sec:future}

\Para{Characterization of $\xrightarrow{p}$}: Unlike Lamport's $\rightarrow$ relation, $\xrightarrow{p}$ relation is not necessarily transitive, which makes extracting the linear ordering a challenge. More research is needed to (i) render $\xrightarrow{p}$ transitive by some transformation of the problem space (e.g., barring the relation of some elements by enforcing constraints on event occurrence pattern), and (ii) studying the probability distributions of clock offsets to establish when $\xrightarrow{p}$ can be safely treated as transitive. 

\Para{Host-network variability: } Jitter in the host's data path can affect an application's access to the local clock as well as the latency of sending out a message. Advancements in low-latency and low-jitter host networking (e.g., DPDK~\cite{dpdk}, XDP~\cite{xdp}, RTOS~\cite{rtos}) have minimized latency variations in the host data path. However, it remains to be studied how low latency variance can be reduced and whether it sets an upper bound on the achievable fairness guarantees. 

\Para{Extension to Fair Total Order}: The proposed sequencer emits batches instead of individual messages. As the batch size can be arbitrarily large, some applications may require emitting individual messages instead of batches. Doing this would require extending the fair partial order to fair total order of messages. Arbitrarily breaking ties on messages of a batch would violate fairness as some clients may always be preferred over others. A random mechanism for breaking ties might be of interest as it would lead to stochastic fairness over a sufficiently long duration.

\Para{Learning Clock Offsets Distributions}: Any clock synchronization protocol gives each client enough information to estimate its offsets distribution. Each synchronization probe may add an offset (w.r.t. to the sequencer's clock) to the clock of a client. Such offsets can be used to estimate the distribution. This mechanism may be too brittle for extraordinary conditions like a part of the data-center experiencing abrupt temperature changes, leading to dramatic clock sync. errors. A robust mechanism for capturing such errors in the respective distributions is needed. Similarly, more research is needed to account for the clock drift errors along with the clock offsets errors in the error distributions. 

\Para{Byzantine Clients}: Byzantine failures further complicate the problem of fair sequencing. A study about achievable fairness guarantees in the presence of Byzantine failures is needed. Pompe~\cite{byantine_ordered_consensus} can serve as a promising starting point for further exploration. In \emph{auction-apps}, clients have an incentive to dictate sequencing of messages e.g., by manipulating the timestamps attached to the messages, as it may translate to monetary benefits e.g., winning trades in a financial exchange. In-depth investigation of security boundaries is needed to make fair sequencing practical. The trust models in \cite{onyx} provide a starting point. 

%% file: conclusion.tex
\section{Conclusion}

We present the problem of fair sequencing and associated challenges which warrant substantial future research. We advocate for utilizing clock offset distributions along with a best effort clock synchronization protocol to construct a pairwise relation, \emph{likely-happened-before}. The proposed relation helps in achieving probabilistic fair ordering of events, useful for an emerging class of applications which require fairly ordering several concurrent events.


%% file: acks.tex
\section{ACKNOWLEDGMENTS}

We thank the reviewers and our shepherd, Jon Crowcroft, for their helpful comments. We thank Ramakrishnan Krishnamurthy and Fabian Ruffy for their helpful comments and productive brainstorming sessions. This work was supported by NSF CAREER award (2340748). 

%% file: appendix.tex
\input{guassian_transitivity_proof}

%% file: guassian_transitivity_proof.tex
\appendix
\section{Transivity holds for Gaussian Distributions}
\label{app:proof_guass}
\begin{proposition} Let $X,Y,Z$ be independent normal random variables
\[
X\sim\mathcal N(\mu_X,\sigma_X^{2}),\qquad
Y\sim\mathcal N(\mu_Y,\sigma_Y^{2}),\qquad
Z\sim\mathcal N(\mu_Z,\sigma_Z^{2}).
\]
Define the preference relation
\[
X \succ Y \;\Longleftrightarrow\; \Pr[X>Y]>\tfrac12.
\]
Then $\succ$ is transitive: if $X\succ Y$ and $Y\succ Z$, we necessarily have $X\succ Z$.
\end{proposition}

\begin{proof}
For any two independent Gaussian variables $A\sim\mathcal N(\mu_A,\sigma_A^{2})$ and
$B\sim\mathcal N(\mu_B,\sigma_B^{2})$, the difference $A-B$ is Gaussian with
\[
A-B\;\sim\;\mathcal N\!\bigl(\mu_A-\mu_B,\;\sigma_A^{2}+\sigma_B^{2}\bigr).
\]
Hence
\[
\Pr[A>B]=\Pr[A-B>0]
        =\Phi\!\Bigl(\frac{\mu_A-\mu_B}{\sqrt{\sigma_A^{2}+\sigma_B^{2}}}\Bigr),
\]
where $\Phi$ is the standard–normal CDF. Now:

\begin{equation}\label{eq:start}
    \Pr[A>B]>\tfrac12
    \quad\Longleftrightarrow\quad
    \Phi\!\Bigl(\frac{\mu_A-\mu_B}{\sqrt{\sigma_A^{2}+\sigma_B^{2}}}\Bigr) >\tfrac12.
\end{equation}

As \(\Phi(0) = \tfrac12\), so:

\[
    \Phi\!\Bigl(\frac{\mu_A-\mu_B}{\sqrt{\sigma_A^{2}+\sigma_B^{2}}}\Bigr) >\Phi(0).
\]

Because $\Phi$ is a strictly increasing function,
\[
\Phi\left( \frac{\mu_A - \mu_B}{\sqrt{\sigma_A^2 + \sigma_B^2}} \right) > \Phi(0)
\quad \Longleftrightarrow \quad
\frac{\mu_A - \mu_B}{\sqrt{\sigma_A^2 + \sigma_B^2}} > 0.
\]

As the denominator $\sqrt{\sigma_A^2 + \sigma_B^2}$ cannot be negative, 

\begin{equation}
    \label{eq:means_matter}
    \Pr[A>B]>\tfrac12
    \quad\Longleftrightarrow\quad
    \mu_A-\mu_B>0
    \quad\Longleftrightarrow\quad
    \mu_A>\mu_B.
\end{equation}

Thus our preference rule depends \emph{only} on the means.

Now, suppose $X\succ Y$ and $Y\succ Z$.  This implies
\[
\mu_X>\mu_Y \quad\text{and}\quad \mu_Y>\mu_Z,
\]
which together give $\mu_X>\mu_Z$ because means (i.e., real numbers) are transitive. Applying eq. \ref{eq:means_matter} to $\mu_X>\mu_Z$, yields
$X\succ Z$.
\end{proof}


\section{Illustrative Example of Fair Ordering}
\label{app:fair_ordering_example}

We now walk through an example that illustrates the probabilistic ordering and batching process described in Section~\ref{batching}. The example involves four messages, $\{A, B, C, D\}$, each carrying a timestamp from a client clock. Because clocks are only approximately synchronized, the sequencer infers pairwise probabilities for which message likely occurred before another. These probabilities are derived from the clock-offset distributions.

\subsection{Constructing the Graph}
Suppose the sequencer estimates the following pairwise probabilities:

\[
\begin{array}{c|cccc}
 & A & B & C & D \\ \hline
A & - & 0.85 & 0.65 & 0.92 \\
B & 0.15 & - & 0.72 & 0.68 \\
C & 0.35 & 0.28 & - & 0.80 \\
D & 0.08 & 0.32 & 0.20 & - \\
\end{array}
\]

Each cell $(i,j)$ represents the probability $p$ that $i \xrightarrow{p} j$, i.e., message $i$ likely precedes message $j$.  
For every unordered pair $(i,j)$, we retain the edge with the higher probability and discard the reverse edge.  
For instance, between $(A,B)$, we keep $A \xrightarrow{0.85} B$ and discard $B \xrightarrow{0.15} A$.

The resulting directed edges form a tournament:
\[
A \xrightarrow{0.85} B,\,
A \xrightarrow{0.65} C,\,
A \xrightarrow{0.92} D,\,
B \xrightarrow{0.72} C,\,
C \xrightarrow{0.80} D,\,
B \xrightarrow{0.68} D.
\]

\subsection{Extracting the Linear Order}
This graph is acyclic and admits a unique topological ordering:
\[
A \prec B \prec C \prec D.
\]
If, however, some edges such as $C \xrightarrow{0.55} A$ were reversed, a cycle ($A \to B \to C \to A$) could form, reflecting an intransitive $\xrightarrow{p}$ relation.  
Breaking such cycles would require edge removals or probabilistic adjustments, which may introduce unfairness—illustrating the complexity discussed in Section~\ref{batching}.

\subsection{Batch Formation}
Even under a transitive ordering, adjacent pairs can differ substantially in confidence.  
Here, $A \xrightarrow{0.85} B$ and $C \xrightarrow{0.80} D$ both have high confidence, while $B \xrightarrow{0.72} C$ is more ambiguous.  
Using ${Threshold}=0.75$, we form a batch boundary wherever $p > 0.75$ between consecutive messages—indicating a confident precedence that warrants separation into distinct batches.

\[
A \xrightarrow{0.85} B \xrightarrow{0.72} C \xrightarrow{0.80} D
\]

Two boundaries are created:  
- one between $A$ and $B$ (since $0.85 > 0.75$), and  
- one between $C$ and $D$ (since $0.80 > 0.75$).

No boundary appears between $B$ and $C$, because their probability $0.72$ is below the threshold, meaning the sequencer cannot confidently distinguish their order.  
The resulting batches are therefore:
\[
\text{Batch}_0 = \{A\}, \quad
\text{Batch}_1 = \{B, C\}, \quad
\text{Batch}_2 = \{D\}.
\]

The sequencer assigns Batch$_0$ rank 0, Batch$_1$ rank 1, and Batch$_2$ rank 2, yielding the final fair ordering:
\[
\{A\} \prec \{B, C\} \prec \{D\}.
\]

\smallskip
A higher threshold (e.g., 0.9) would result in fewer, larger batches—indicating stricter confidence requirements—while a lower threshold (e.g., 0.6) would yield finer-grained batching, approaching a total order.  
This example demonstrates how probabilistic confidence directly controls the granularity of fair ordering.


\section{Illustrative Example of Online Sequencing}
\label{app:online_seq_example}

We now provide an example corresponding to the discussion in Section~\ref{sec:streaming}. The example demonstrates how the sequencer answers the two key questions: ensuring all relevant messages have arrived (Q2) and determining 
how much to wait for new messages before emitting a batch of messages(Q1).

\subsection*{Q2: Ensuring Completeness of Message Arrivals}

Consider two clients, $C_1$ and $C_2$, each continuously sending messages to the sequencer with monotonically increasing local timestamps.  
Because network delays may differ across clients, messages do not necessarily arrive in timestamp order. The sequencer must ensure that when it emits a batch containing all messages up to timestamp $t$, no message with a timestamp smaller than $t$ is still in flight.

Assuming the sequencer knows the complete set of participating clients, a simple and robust rule suffices:  
the sequencer waits until it has received a message or heartbeat from \emph{each client} carrying a timestamp greater than $t$.  
Once this condition holds, it can safely conclude that all messages with timestamps $\leq t$ have already arrived.

This mechanism works regardless of variable network delay, as long as each client communicates through an ordered delivery channel (e.g., a TCP connection). It effectively bounds asynchrony and guarantees that the sequencer does not emit a batch prematurely.

\subsection*{Q1: What future messages may need to be included in a given batch of messages?}

We now examine how the sequencer determines which future messages might still need to be included in a given batch before emitting it.  
This question arises from clock uncertainty: even if two messages appear temporally separated in their local timestamps, their offsets distributions may overlap enough with the distribution of another client, forcing the sequencer to group multiple messages of one client together with the message of another client.

Assume there are two clients, $C_1$ and $C_2$, each with slightly different clock offsets.  
Client $C_1$ sends two messages ($1a$ and $1b$), while $C_2$ sends one message ($2$).  
The true (global) generation times are:
\[
T^*_{1a}=100.0, \quad T^*_{2}=100.2, \quad T^*_{1b}=100.3.
\]
Client $C_2$’s clock, however, is significantly more uncertain than $C_1$’s.  
Due to these offsets, the sequencer receives the reported timestamps as:
\[
t_{1a}=100.0, \quad t_{2}=100.6, \quad t_{1b}=100.3,
\]
and the messages arrive in the order $t_{1a} \rightarrow t_{2} \rightarrow t_{1b}$.

\smallskip
\noindent\textbf{Step 1: Initial batching.}
When $C_1$’s first message ($1a$) arrives, it forms its own tentative batch:
\[
\text{Batch}_0 = \{1a\}.
\]

The sequencer cannot emit a batch until it has met the criteria for safe emission, i.e., it has waited enough time so that no new messages can arrive that may belong to the same batch. We will visit the safe emission later in the example, assume for now that new messages arrive before safe emission criteria is met. 

\smallskip
\noindent\textbf{Step 2: Arrival of a high-uncertainty message.}
When $C_2$’s message arrives with timestamp $t_2 = 100.6$, its wide uncertainty interval means the sequencer cannot rule out the possibility, based on preceding probabilities, that it occurred before or after $1a$ in global time.  
To preserve fairness, the sequencer merges the two into one batch:
\[
\text{Batch}_0 = \{1a, 2\}.
\]
The batch remains \emph{open}, since a future message might still belong to it.

\smallskip
\noindent\textbf{Step 3: Arrival of a later message from the same client.}
Soon after, $C_1$ sends another message ($1b$) with timestamp $t_{1b}=100.3$.  
Even though $1b$ clearly follows $1a$ locally, the uncertainty around $C_2$’s message makes it impossible to confidently separate $1b$ from the ongoing batch.  
Hence, to maintain fairness, the sequencer places it in the same batch:
\[
\text{Batch}_0 = \{1a, 1b, 2\}.
\]

\smallskip
\noindent\textbf{Step 4: Safe emission.}
The sequencer computes for each message $i$ a future time $T^F_i$ such that
\[
\mathbb{P}(T^*_i < T^F_i) > p_{\mathrm{safe}},
\]
and defines the safe emission time of the batch as:
\[
T_b = \max_{k \in \text{Batch}_0} T^F_k.
\]
Once the sequencer’s clock reaches $T_b$ and no new message has arrived that belongs to Batch$_0$ (based on preceding probabilities), then the batch is considered safe to be emitted i.e., it is very unlikely that a new message will arrive that needs to belong to the batch being emitted. 

\smallskip
\noindent\textbf{Discussion.}
This example illustrates that a single high-uncertainty message (here, from $C_2$) can force multiple temporally distinct messages from another client (here, $C_1$’s $1a$ and $1b$) to share the same batch.  
The sequencer’s decision therefore depends not only on per-client timestamp order but also on the joint uncertainty distribution across clients.  
The choice of $p_{\mathrm{safe}}$ determines the trade-off between fairness confidence and emission latency.